\documentclass[aps,reprint,floatfix]{revtex4-1}
\usepackage{amsmath,amssymb,graphicx,color,mathtools,amsthm}
\usepackage{calligra,array}
\usepackage{blkarray}
\usepackage{multirow}
\usepackage{hhline}
\usepackage{slashbox}

\newcounter{theoremcounter}
\stepcounter{theoremcounter}

\newtheorem{theorem}{Theorem}
\newtheorem{lemma}{Lemma}

\theoremstyle{definition}
\newtheorem{definition}{Definition}

\newcommand{\bs}{\boldsymbol}
\newcommand{\bra}[1]{\left\langle #1\right|}
\newcommand{\ket}[1]{\left|#1\right\rangle}

\newcommand{\ketbra}[2]{\ket{#1}\bra{#2}}
\newcommand{\partder}[2]{\frac{\partial #1}{\partial #2}}

\makeatletter
\newcommand{\vast}{\bBigg@{4}}
\newcommand{\Vast}{\bBigg@{5}}
\makeatother

\makeatletter
\newcommand{\thickhline}{%
    \noalign {\ifnum 0=`}\fi \hrule height 1pt
    \futurelet \reserved@a \@xhline
}
\newcolumntype{;}{@{\hskip\tabcolsep\vrule width 1pt\hskip\tabcolsep}}
\makeatother

\begin{document}

\title{Discrete Wigner Function Derivation of the Aaronson-Gottesman Tableau Algorithm}
\author{Lucas Kocia}
\affiliation{Department of Physics, Tufts University, Medford, Massachusetts 02155, U.S.A.}
\author{Yifei Huang}
\affiliation{Department of Physics, Tufts University, Medford, Massachusetts 02155, U.S.A.}
\author{Peter Love}
\affiliation{Department of Physics, Tufts University, Medford, Massachusetts 02155, U.S.A.}
\begin{abstract}
The Gottesman-Knill theorem established that stabilizer states and operations can be efficiently simulated classically. For qudits with dimension three and greater, stabilizer states and Clifford operations have been found to correspond to positive discrete Wigner functions and dynamics. We present a discrete Wigner function-based simulation algorithm for odd-$d$ qudits that has the same time and space complexity as the Aaronson-Gottesman algorithm. We show that the efficiency of both algorithms is due to the harmonic evolution in the symplectic structure of discrete phase space. The differences between the Wigner function algorithm and Aaronson-Gottesman are likely due only to the fact that the Weyl-Heisenberg group is not in $SU(d)$ for $d=2$ and that qubits have state-independent contextuality. This may provide a guide for extending the discrete Wigner function approach to qubits.
\end{abstract}
\maketitle

\section{Introduction}
\label{sec:intro}

A general set of unitary quantum operators on \(n\)-qudit states requires a number of operators that grows exponentially with \(n\). An important exception to this rule involves the set of Clifford operators acting on stabilizer states. These states play an important role in quantum error correction~\cite{Gottesman98} and are closed under action by Clifford gates. Efficient simulation of such systems was demonstrated with the tableau algorithm of Aaronson and Gottesman~\cite{Aaronson04,Gottesman98} for qubits (\(d=2\)). Finding the underlying reason for why such an efficient algorithm is possible for Clifford circuit simulation has since been the subject of much study~\cite{Gottesman99,Mari12,Howard14}.

Recent progress has been the result of work by Wootters~\cite{Wootters87}, Eisert~\cite{Mari12}, Gross~\cite{Gross06}, and Emerson~\cite{Howard14}, who have formulated a new perspective based on the discrete phase spaces of states and operators in finite Hilbert spaces using discrete Wigner functions. They have shown that stabilizer states have positive-definite discrete Wigner functions and that Clifford operators are positive-definite maps. This implies that Clifford circuits are efficiently simulatable on classical computers. In odd-dimensional systems, stabilizer states have been shown to be the discrete analogue to Gaussian states in continuous systems~\cite{Gross06} and Clifford group gates have been shown to have underlying harmonic Hamiltonians that preserve the discrete Weyl phase space points~\cite{Kocia16}. This means Clifford circuits are expressible by path integrals truncated at order \(\hbar^0\) and are thus manifestly classical~\cite{Kocia16,Dax17}.

This poses the question of what the relationship is between past efficient algorithms for Clifford circuits and the propagation of discrete Wigner functions of stabilizer states under Clifford operators. Here, we show that the original Aaronson-Gottesman tableau algorithm is actually equivalent to such a discrete Wigner function propagation and that the tableau matrix coincides with the discrete Wigner function of a stabilizer state. We accomplish this by first developing a Wigner function-based algorithm that classically simulates stabilizer state evolution under Clifford gates and measurements in the \(\hat Z\) Pauli basis for odd \(d\). We then show its equivalence to the well-known Aaronson-Gottesman tableau algorithm~\cite{Aaronson04}. Both algorithms require \(\mathcal{O}(n^2)\) qudits to represent \(n\) stabilizer states, \(\mathcal{O}(n)\) operations per Clifford operator, and deterministic and random measurements require \(\mathcal{O}(n^2)\) operations.

The Aaronson-Gottesman tableau algorithm makes use of the Heisenberg representation. This means that evolution is accomplished by updating an associated tableau or matrix representation of the Clifford operators instead of the stabilizer states themselves. Our algorithm is framed in the Schr\"odinger picture, and involves evolving the Wigner function of stabilizer states. Since the two algorithms are equivalent, the formulation of Clifford simulation in the Heisenberg picture is a choice and not a necessity for its efficient simulation. Furthermore, we reveal the purely classical basis of both algorithms and the physically intuitive phase space structures and symplectic properties on which they rely.

\section{Discrete Wigner Function of the Stabilizer Formalism}
\label{sec:stabcenterchord}

Before we discuss the discrete Wigner function, we introduce a basic framework that defines how a phase space behaves for odd \(d\)-dimensional Hilbert spaces. To begin, we associate the computational basis with the position basis, such that the Pauli \(\hat Z_j\) operator on the \(j\)th qubit for \(n\) qubits acts as a ``boost'' operator:
\begin{equation}
  \hat Z_j \ket{k_1, \ldots, k_j, \ldots, k_n} = e^{\frac{2 \pi i}{d} k_j} \ket{k_1, \ldots, k_j, \ldots, k_n}.
\end{equation}

The discrete Fourier transform operator is defined by:
\begin{eqnarray}
  \label{eq:hadamard}
  &&\hat F_j =\\
  && \frac{1}{\sqrt{d}} \sum_{\substack{k_j,l_j \in \\\mathbb{Z}/d\mathbb{Z}}} e^{\frac{2 \pi i}{d} k_j l_j} \ketbra{k_1,\ldots,k_j,\ldots,k_n}{l_1,\ldots,l_j,\ldots,\l_n}, \nonumber
\end{eqnarray}
This is equivalent to the Hadamard gate and allows us to define the Pauli \(\hat X_j\) operator as follows:
\begin{equation}
   \hat X_j \equiv \hat F_j^\dagger \hat Z_j \hat F_j.
\end{equation}

While \(\hat Z_j\) is a boost, \(\hat X_j\) is a shift operator because
\begin{equation}
  \hat X_j^{\delta q} \ket{k_1,\ldots,k_j,\ldots,k_n} \equiv \ket{k_1,\ldots,k_j\oplus\delta q,\ldots,k_n},
\end{equation}
where \(\oplus\) denotes integer addition mod \(d\).

We can reexpress the boost \(\hat Z_j\) and shift \(\hat X_j\) operators in terms of their generators, which are the conjugate \(\hat q_j\) and \(\hat p_j\) operators respectively:
\begin{equation}
  \label{eq:boost}
  \hat Z_j = e^{\frac{2 \pi i }{d}\hat q_j}
\end{equation}
and
\begin{equation}
  \label{eq:shift}
  \hat X_j = e^{-\frac{2 \pi i }{d} \hat p_j}.
\end{equation}
Thus, we can refer to the \(\hat X_j\) basis as the momentum (\(p_j\)) basis, which is equivalent to the Fourier transform of the \(q_j\) basis:
\begin{equation}
  \hat p_j = \hat F_j^\dagger \hat q_j \hat F_j.
\end{equation}
These bases form the discrete Weyl phase space \((\bs p,\bs q)\).

The Wigner function \(\Psi_x(\bs p,\bs q)\) of a pure state \(\ket{\Psi}\) is defined on this discrete Weyl phase space:
\begin{widetext}
\begin{eqnarray}
  \label{eq:weylfunctionpurestatediscrete}
  \Psi_x(\bs p,\bs q) &=& d^{-n} \sum_{\substack{\bs \xi_q \in \\(\mathbb{Z} / d \mathbb{Z})^{ n}}} e^{-\frac{2 \pi i}{d} \bs \xi_q \cdot \bs p} \Psi \left( \bs q + \frac{(d+1) \bs \xi_q }{2} \right) {\Psi^*} \left( \bs q - \frac{(d+1) \bs \xi_q}{2} \right).
\end{eqnarray}
\end{widetext}
We will shortly be interested in the discrete Wigner function of stabilizer states. But first we introduce the effect that the Clifford gates have in this discrete Weyl phase space.

\subsection{Clifford Gates}
A Clifford group gate \(\hat V\) is related to a symplectic transformation on the discrete Weyl phase space, governed by a symplectic matrix \(\bs{\mathcal M_{\hat V}}\) and vector \(\bs \alpha_{\hat V}\)~\cite{Kocia16}:
\begin{equation}
  \left( \begin{array}{c}\bs p'\\ \bs q'\end{array} \right) = \bs{\mathcal M}_{\hat V} \left[ \left( \begin{array}{c}\bs p\\ \bs q \end{array}\right) + \frac{1}{2} \bs \alpha_{\hat V} \right] + \frac{1}{2} \bs \alpha_{\hat V}.
  \label{eq:quadmap}
\end{equation}
Wigner functions \(\Psi_x(\bs x)\) of states evolve under Clifford operators \(\hat V\) by
\begin{equation}
  \Psi_x\left(\bs{\mathcal M_{\hat V}} \left(\bs x + \bs \alpha_{\hat V}/2\right) + \bs \alpha_{\hat V}/2\right),
  \label{eq:wignerevolution}
\end{equation}
where \(\bs x \equiv (\bs p, \bs q)\).

When considering Clifford gate propagation, we can restrict to the generators of the group. One such set of generators is made up of the phase-shift gate \(\hat P_i\), the Hadamard gate \(\hat F_i\), and the controlled-not \(\hat C_{ij}\) (which act on the \(i\)th and \(j\)th qudits).

The phase shift \(\hat P_i\) is a one-qudit gate with the underlying Hamiltonian \(H_{\hat P_i} = -\frac{d+1}{2} q_i^2 + \frac{d+1}{2} q_i\)~\cite{Kocia16}. Without loss of generality, we will instead consider
\begin{equation}
  \hat P'_i = \hat P_i \hat P_i \hat Z_i,
\end{equation}
which we will refer to as the phase-shift gate in this paper.

We note that the usual phase-shift can be obtained from the new one within the Clifford group:
\begin{equation}
  \hat P_i = \hat P'_i \hat P'_i \hat Z_i,
\end{equation}
where \([\hat P_i, \hat Z_i] = [\hat P'_i, \hat Z_i]=0\). Hence, \(\hat P'_i\) is an adequate replacement generator for \(\hat P_i\), and we will use it instead of \(\hat P_i\) from now on. Since its Hamiltonian has no linear term (\(H_{\hat P'_i} = -q_i^2\)), this leads to an easier presentation ahead since \(\bs \alpha_{\hat P'_i} = \bs 0\).

The corresponding equations of motion for \(\hat P'_i\) are \(\dot p_i = 2 q_i\) and \(\dot q_i = 0\). Hence, for \(\Delta t = 1\),
\begin{equation}
  \label{eq:phaseshiftstabmat}
  \left( \bs{\mathcal M}_{\hat P'_i} \right)_{j,k} = \delta_{j,k} + 2 \delta_{i,j} \delta_{n+i,k}.
\end{equation}

The Hadamard gate \(\hat F_i\) is a one-qudit gate and has the underlying Hamiltonian \(H_{\hat F_i} = -\frac{\pi}{4} (p_i^2 + q_i^2)\)~\cite{Kocia16}. The corresponding equations of motion are \(\dot p_i = \frac{\pi}{2} q_i\) and \(\dot q_i = -\frac{\pi}{2} p_i\). Hence, for \(\Delta t = 1\),
\begin{eqnarray}
  \label{eq:hadamardstabmat}
  \left(\bs{\mathcal M}_{\hat F_i}\right)_{j,k} = \delta_{j,k} - \delta_{i,j} \delta_{i,k} - \delta_{n+i,j} \delta_{n+i, k}\\
  + \delta_{i,j} \delta_{n+i,k} - \delta_{n+i,j} \delta_{i,k}, \nonumber
\end{eqnarray}
and \(\bs \alpha_{\hat F_i} = \bs 0\).

Finally, the two-qudit CNOT \(\hat C_{ij}\) on control qudit \(i\) and second qudit \(j\) has the corresponding Hamiltonian \(H_{\hat C_{ij}} = p_i q_j\)~\cite{Kocia16}. The corresponding equations of motion are \((\dot p_i,\dot p_j) = -(0,p_i)\) and \((\dot q_i, \dot q_j) = (q_j,0)\). Hence, for \(\Delta t = 1\),
\begin{equation}
  \label{eq:cnotstabmat}
\left(\bs{\mathcal M}_{\hat C_{ij}}\right)_{k,l} = \delta_{k,l} - \delta_{i,k} \delta_{j,l} + \delta_{n+j,k} \delta_{n+i,l},
\end{equation}
and \(\bs \alpha_{\hat C_{ij}} = \bs 0\).

\subsection{Wigner Functions of Stabilizer States}

A discrete Wigner function associated with the boost and shift operators defined in Eqs.~\ref{eq:boost} and~\ref{eq:shift} is given by the following theorem~\cite{Kocia16}:
  \begin{theorem}
  \label{thm:wignerfnofstabstateoddd}
    The discrete Wigner function \(\Psi_x(\bs x)\) of a stabilizer state \(\Psi\) for any odd \(d\) and \(n\) qudits is \(\delta_{\bs \Phi \times \bs x, \bs r}\) for \(2n \times 2n\) matrix \(\bs \Phi\) and \(2n\) vector \(\bs r\) with entries in \(\mathbb{Z}/d\mathbb{Z}\).
  \end{theorem}
  In particular, if we begin with a stabilizer state defined as \(\ket \Psi_0 = \ket {\bs q_0}\), then it was shown in~\cite{Kocia16} that \({\Psi_x}(\bs x) = \delta_{\bs \Phi_0 \times \bs x, \bs r_0}\), where \(\bs \Phi_0 = \left(\begin{array}{cc} 0 & 0\\ 0 & \mathbb{I}_n\end{array} \right)\), and \(\bs r_0 = (\bs 0, \bs q_0)\).

  \section{Wigner Stabilizer Algorithm}
  \label{sec:wigstabalgorithm}

  With the discrete Wigner function of a stabilizer state defined in Theorem~\ref{thm:wignerfnofstabstateoddd} and the effect of the Clifford group generators on discrete Wigner functions defined in Eq.~\ref{eq:wignerevolution}, we can now examine the effect Clifford operators have on stabilizer states.
  
\subsection{Stabilizer Representation}
\label{sec:stabrep}

From Theorem~\ref{thm:wignerfnofstabstateoddd}, propagation of the stabilizer state \(\Psi\) can be represented by considering the state in the Wigner representation \(\Psi_x(\bs x) = \delta_{\bs \Phi_t \cdot \bs x, \bs r_t}\). In this way, \(\bs \Phi_t\) and \(\bs r_t\) specify a linear system of equations in terms of \(\bs p_t\) and \(\bs q_t\). The first \(n\) rows of \(\bs \Phi_t\) are the coefficients of \((\bs p_t, \bs q_t)^T\) in \(\bs p_0(\bs p_t, \bs q_t)\) and the last \(n\) rows of \(\bs \Phi_t\) are the coefficients of \((\bs p_t, \bs q_t)^T\) in \(\bs q_0(\bs p_t, \bs q_t)\):
\begin{equation}
  \left( \begin{array}{c} \bs p_0\\ \bs q_0 \end{array} \right) = \bs \Phi_t \left( \begin{array}{c} \bs p_t\\ \bs q_t \end{array} \right),
\end{equation}
The Kronecker delta function sets this linear system of equations equal to \(\bs r_t\). In this way, an affine map---a linear transformation displaced from the origin by \(\bs r_t\)---is defined. This system of equations must be updated after every unitary propagation and measurement.

Since the Wigner functions \(\Psi_x(\bs x)\) of stabilizer states propagate under \(\bs{\mathcal M}\) as \(\Psi_x(\bs{\mathcal M \bs x})\), it follows that
\begin{equation}
  \label{eq:phiprop}
  \bs \Phi_t \rightarrow \bs \Phi_t \bs{\mathcal M}_t^{-1}.
\end{equation}
The importance of vector \(\bs r_t\) and when it must be updated will become evident when we consider random measurements. Hence, after \(n\) operations \(\bs {\mathcal M}_1\), \(\bs{\mathcal M}_2\), \(\ldots\), \(\bs{\mathcal M}_n\),
\begin{equation}
  \bs{\mathcal M}_t^{-1} = \bs{\mathcal M}_1^{-1} \bs{\mathcal M}_2^{-1} \ldots \bs{\mathcal M}_n^{-1}.
\end{equation}
The matrices are ordered chronologically left-to-right instead of right-to-left.

Since \(\bs{\mathcal M}\) is symplectic, \(\bs{\mathcal M}_t^{-1} = - \bs{\mathcal J} \bs{\mathcal M}_t^T \bs{\mathcal J}\) where
\begin{equation}
  \bs{\mathcal J} = \left(\begin{array}{cc} 0 & -\mathbb{I}_n\\ \mathbb{I}_n & 0 \end{array}\right),
\end{equation}
and \(\mathbb{I}_n\) is the \(n \times n\) identity matrix.

Thus the the stability matrices \(\bs{\mathcal M}\) for \(\hat F_i\), \(\hat P'_i\) and \(\hat C_{ij}\) given in Eqs.~\ref{eq:phaseshiftstabmat}-\ref{eq:cnotstabmat} differ from their inverses only by sign changes in their off-diagonal elements:
\begin{equation}
  \label{eq:phaseshiftstabmatinv}
  \left( \bs{\mathcal M}_{\hat P'_i}^{-1} \right)_{j,k} = \delta_{j,k} - 2 \delta_{i,j} \delta_{n+i,k},
\end{equation}
\begin{eqnarray}
  \label{eq:hadamardstabmatinv}
  \left(\bs{\mathcal M}_{\hat F_i}^{-1}\right)_{j,k} = \delta_{j,k} - \delta_{i,j} \delta_{i,k} - \delta_{n+i,j} \delta_{n+i, k}\\
  - \delta_{i,j} \delta_{n+i,k} + \delta_{n+i,j} \delta_{i,k}, \nonumber
\end{eqnarray}
and
\begin{equation}
  \label{eq:cnotstabmatinv}
\left(\bs{\mathcal M}_{\hat C_{ij}}^{-1}\right)_{k,l} = \delta_{k,l} + \delta_{i,k} \delta_{j,l} - \delta_{n+j,k} \delta_{n+i,l}.
\end{equation}

We assume the quantum state is initialized in the computational basis state \(\Psi_0 = \underbrace{\ket{0} \otimes \cdots \otimes \ket{0}}_{n}\) and so initially we should set \(\bs \Phi_0 = \left(\begin{array}{cc} \bs 0 & \bs 0\\ \bs 0 & \mathbb{I}_n \end{array}\right)\) and \(\bs r_0 = \bs 0\). The initial stabilizer state is \(\bs {\Psi_x}_0 = \delta_{\bs q_t, 0}\). However, it will become clear when we discuss measurements that it is practically useful to instead set
\begin{equation}
  \bs \Phi_0 = \left(\begin{array}{cc} \mathbb{I}_n & \bs 0\\ \bs 0 & \mathbb{I}_n \end{array}\right),
\end{equation}
thereby setting \({\Psi_x}_0 = \delta_{(\bs p_t, \bs q_t), (\bs 0, \bs 0)}\)---not a true Wigner function. This is equivalent to the last matrix if the first \(n\) rows in \(\bs \Phi_t \bs x\) and \(\bs r_t\) are ignored---the same as ignoring \(\bs p_0(\bs p_t, \bs q_t)\). In fact, we have two Wigner functions here: one defined by the first \(n\) rows and another by the last \(n\) rows. We proceed in this manner, ignoring the first \(n\) rows, until their usefulness becomes apparent to us. 

For \(n\) qudits unitary propagation requires \(\mathcal O(n^2)\) dits of storage to track \(\bs{\Phi}_t\) and \(\bs r_t\). More precisely, since \(\bs{\Phi}_t\) is a \(2n\times 2n\) matrix and \(\bs r_t\) is an \(2n\)-vector, \(2n(2n+1)\) dits of storage are necessary.

\subsection{Unitary Propagation}
\label{sec:propalgorithm}

\(\bs \Phi_t\) contains the coefficients of the linear equations relating \(\bs x_0\) to \(\bs x_t\). Each row is one equation relating \({q_0}_i\) or \({p_0}_i\) to \(\bs x_t\). When manipulating rows of \(\bs \Phi_t\) we shall refer to the linear equations that these rows define.

Examining Eqs.~\ref{eq:phaseshiftstabmatinv},~\ref{eq:hadamardstabmatinv}, and~\ref{eq:cnotstabmatinv}, we see that the inverse stability matrices of the generator gates \(\hat F_i\), \(\hat P_i\) and \(\hat C_{ij}\) are the sum of an identity matrix and a matrix with a finite number of non-zero off-diagonal elements. The number of these off-diagonal elements is independent of the number of qudits, \(n\). Hence, multiplying \(\bs \Phi_t\) with a new stability matrix in Eq.~\ref{eq:phiprop} and evaluating the matrix multiplication is equivalent to performing a finite number of \(n\)-vector dot products and so requires \(\mathcal O(n)\) operations. Therefore, keeping track of propagation of stabilizer states by Clifford gates can be simulated with \(\mathcal O(n)\) operations.

Let us examine these unitary operations more closely. Defining \(\oplus\) and \(\ominus\) to be mod \(d\) addition and subtraction respectively, we find:

\emph{Phase gate on qudit \(i\) (\(\hat P'_i\))}. For all \(j \in \{1,\, \ldots,\, 2n\}\), set \(\bs{\Phi}_{j,n+i} \mapsto \bs{\Phi}_{j,n+i} \ominus 2 \bs{\Phi}_{j,n}\).

\emph{Hadamard gate on qudit \(i\) (\(\hat F_i\))}. Negate \(\bs{\Phi}_{j,i}\) mod \(d\), and then swap \(\ominus \bs{\Phi}_{j,i}\) and \(\bs{\Phi}_{n+i,j}\).

\emph{CNOT from control \(i\) to target \(j\) (\(\hat C_{ij}\))}. Set \(\bs{\Phi}_{k,j} \mapsto \bs{\Phi}_{k,j} \oplus \bs{\Phi}_{k,i}\) and \(\bs{\Phi}_{k,n+i} \mapsto \bs{\Phi}_{k,n+i} \ominus \bs{\Phi}_{k,n+j}\).

This confirms that unitary propagation in this scheme requires \(\mathcal O(n)\) operations.

\subsection{Measurement}
\label{sec:measalgorithm}

The outcome of a measurement \(\hat Z_i\) on a stabilizer state can be either random or deterministic. As described above, the bottom half of \(\bs \Phi_t\) defines \({q_0}_j\) for \(j \in \{1, \,\ldots,\, n\}\), each of which is a linear combination of \({q_t}_i\) and \({p_t}_i\). The entries in the \((n+j)\)th row of \(\bs \Phi_t\) give the coefficient of \({p_t}_i\) and \({q_t}_i\) in \({q_0}_j\) for \(j \in \{1, \,\ldots,\, n\}\). If the coefficient of \({p_t}_i\) in any \({q_0}_j\) is non-zero then the measurement \(\hat Z_i\) will be random. If all coefficients of \({p_t}_i\) are zero for \({q_0}_j\) \(\forall j\), then the measurement of \(\hat Z_i\) will be deterministic. This can be seen from the fact that if our stabilizer state \(\ket \Psi\) is an eigenstate of \(\hat Z_i\) then  \(\hat Z_i \ket \Psi = e^{i \phi} \ket \Psi\) for some \(\phi \in \mathbb R\) and (discrete) Wigner functions do not change under a global phase. Thus, measuring \(\hat Z_i\) leaves the Wigner function of \(\ket \Psi\) invariant if the measurement is deterministic. Since \(\hat Z_i\) is a boost operator that increments the momentum of a state by one, its effect on the linear system of equations specified by the Wigner function is:
\begin{equation}
  \bs \Phi_t \left(\begin{array}{c}{p_t}_1\\ \vdots\\ {p_t}_i\\ \vdots\\ {p_t}_n\\ \bs q_t \end{array}\right) = \left(\begin{array}{c}\bs {r_t}_p\\ \bs {r_t}_q\end{array}\right) \underset{\hat Z_i}{\mapsto} \bs \Phi_t \left(\begin{array}{c}{p_t}_1\\ \vdots\\ {p_t}_i+1\\ \vdots\\ {p_t}_n\\ \bs q_t \end{array}\right) = \left(\begin{array}{c}\bs {r_t}_p\\ \bs {r_t}_q\end{array}\right).
\end{equation}
Thus, if the lower half of the \(i\)th column of \(\bs \Phi_t\) is zero then \(\hat Z_i\) leaves the Wigner function invariant (and so the measurement is deterministic). Verifying that these coefficients are all zero takes \(\mathcal O(n)\) operations for each \(\hat Z_i\).

In other words, to see if a given measurement of \(\hat Z_i\) is random or deterministic, a search must be performed for non-zero \({\bs{\Phi}_t}_{n+j,i}\) elements. If such a non-zero element exists then the measurement is random since it means that the final momentum of qudit \(i\) affects the state of the stabilizer and so its position must be undetermined (by Heisenberg's uncertainty principle). If no such finite \({\bs{\Phi}_t}_{n+j,i}\) element exists, then the measurement \(\hat Z_i\) is deterministic.

We now describe the algorithm in detail for these two cases:

\textbf{Case 1:} Random Measurement

Let the \((n+j)\)th row in the bottom half of \(\bs \Phi_t\) have a non-zero entry in the \(i\)th column, \({\bs{\Phi}_t}_{n+j,i}\). Since the random measurement \(\hat Z_i\) will project qudit \(i\) onto a position state, we will replace the \((n+j)\)th row with \({q_0}_i = q'_i\) (the uniformly random outcome of this measurement). After this projection onto a position state, none of the other qudit's positions should depend on qudit \(i\)'s momentum, \({p_t}_i\). To accomplish this, before we replace row \((n+j)\), we solve its equation for \({p_t}_i\) and substitute every instance of \({p_t}_i\) in the linear system of equations with this solution. As a result, every equation will no longer depend on \({p_t}_i\) and we can go ahead and replace the \((n+j)\)th row with \({q_0}_i = q'_i\).

There is one more thing to do, which will be important for deterministic measurements: replace the \(j\)th row with the old \((n+j)\)th row. This sets \({p_0}_i = {q_0}_j(\bs p_t, \bs q_t)\), which becomes the only remaining equation explicitly dependent on \({p_t}_i\). In other words, \({p_0}_i \propto {p_t}_i\), similar to the beginning when we set \({p_0}_i = {p_t}_i\) by setting \(\bs \Phi = \mathbb I_{2n}\). However, now we also conserve any dependence \({p_0}_i\) has on the other qudits incurred during unitary propagation. In other words, we conserve \({p_t}_i\)'s dependence upon the other qudits, but only in the Wigner function specified by the top \(n\) rows, which we ignore otherwise.

After replacing the equation specified by row \((n+j)\) of \(\bs{\Phi}_t\) and \(\bs r_t\) with a randomly chosen measurement outcome \(q'_i\) (i.e. \({q_0}_i = q'_i\)), the identification of rows \((n+i)\) and \((n+j)\) are exchanged, so that the former now specifies \({q_0}_j(\bs p_t, \bs q_t)\) while the latter specifies \({q_0}_i(\bs p_t, \bs q_t)\), respectively. \({p_0}_i\) has also been updated by replacing the \(j\)th row in the first half of \(\bs{\Phi}_t\), with the \((n+j)\)th row we just changed. Again, this row now describes \({p_0}_i(\bs p_t, \bs q_t)\) while the \(i\)th row now specifies \({p_0}_j(\bs p_t, \bs q_t)\). Overall, this takes \(\mathcal O(n^2)\) operations since we are replacing \(\mathcal O(n)\) rows with \(\mathcal O(n)\) entries.

\textbf{Case 2:} Deterministic Measurement

Since the measurement is deterministic, \(\bs \Phi_t\) and \(\bs r_t\) do not change. The \(n\) equations specified by the bottom half of \(\bs \Phi_t \bs x_t = \bs r_t\) can be used to solve for \({q_t}_i\)---the deterministic measurement outcome. In general, this can also be done by inverting \(\bs \Psi_t\) and evaluating \(\bs x_t = \bs \Phi_t^{-1} \cdot \bs r_t\) for \(q_i\). Aaronson and Gottesman themselves noted that such a matrix inversion is possible, but practically takes \(\mathcal{O}(n^3)\) operations~\footnote{However, we are not certain if Aaronson and Gottesman were referring to the \(\bs \Phi_t\) matrix corresponding to the \(2n \times 2n\) part of their tableau when they discuss matrix inversion in~\cite{Aaronson04}.}.

Fortunately, there is another method that scales as \(\mathcal O(n^2)\) and requires use of the \(n\) equations represented by the top \(n\) rows of \(\bs \Phi_t\), which were included in our description by setting \(\bs \Phi_0 = \mathbb{I}_{2n}\). The linear system of \(n\) equations represented by \(\bs \Phi_t \bs x_t = \bs r_t\) can be written as
\begin{eqnarray}
  \bs \Phi_t \bs x_t &=& \bs r_t\\
  \left( \begin{array}{c} \bs p_0(\bs p_t, \bs q_t) \\ \bs q_0(\bs p_t, \bs q_t) \end{array} \right) &=& \left( \begin{array}{c} \bs {r_t}_p \\ \bs {r_t}_q \end{array} \right),
\end{eqnarray}
where we are interested in linear combinations of the bottom half, \(\bs q_0(\bs p_t, \bs q_t)\), to solve for the measurement outcome \({q_t}_i\):
\begin{equation}
  \label{eq:lincombdetermistic}
  \sum_{j=1}^n c_{ij} {q_0}_j = {q_t}_i.
\end{equation}
where \(c_{ij} \in \mathbb{Z}/d\mathbb{Z}\).

\begin{lemma}
  \label{lem:measurement}
  The coefficient in front of \({p_t}_i\) in the row of \(\bs \Phi_t\) that specifies \({p_0}_j(\bs p_t, \bs q_t)\), \({\bs \Phi_t}_{ji}\), is equal to the coefficient \(c_{ij}\) in front of \({q_0}_j\) that makes up \({q_t}_i\) in Eq.~\ref{eq:lincombdetermistic}. Or, equivalently,
  \begin{equation}
    \label{eq:deterministiccoeff}
  c_{ij} = {q_0}_j \cdot {q_t}_i(\bs p_0, \bs q_0) = {p_0}_j(\bs p_t, \bs q_t) \cdot {p_t}_i = {\bs \Phi_t}_{ji}.
\end{equation}
\end{lemma}
\begin{proof}
  \begin{figure}[ht]
\includegraphics[scale=1.0]{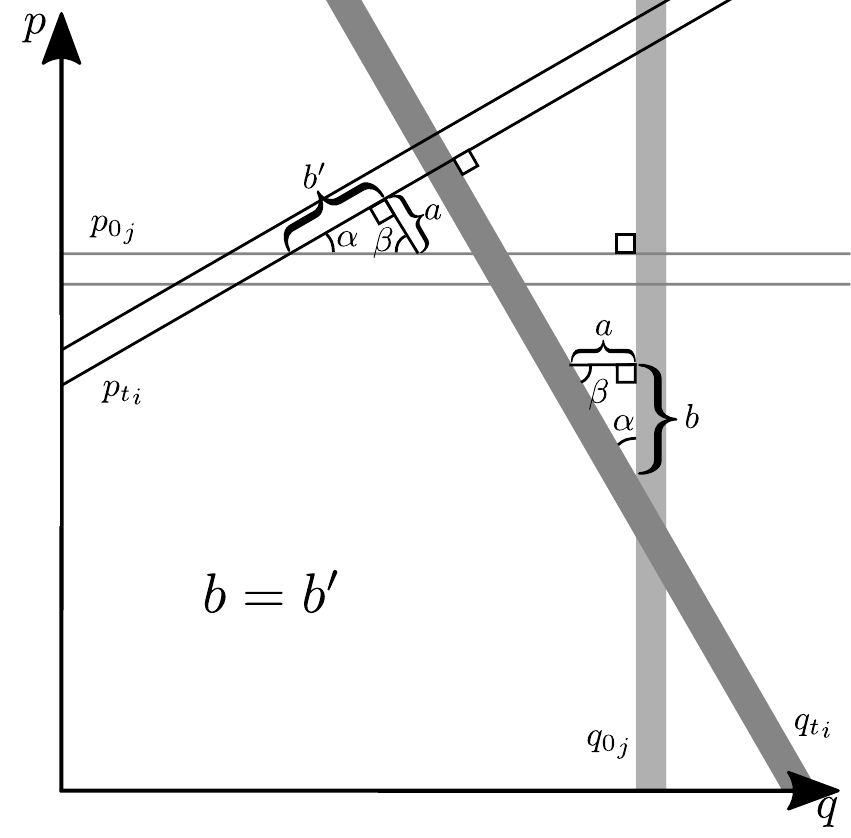}
\caption{The initial perpendicular manifolds \({p_0}_j\) and \({q_0}_j\) and the harmonically evolved perpendicular manifolds \({p_t}_i\) and \({q_t}_i\). Description of the various lengths and angles are given in the text in the proof of Lemma~\ref{lem:measurement}.}
\label{fig:qdotproduct}
  \end{figure}
  Under evolution under the Clifford group operators,
  \begin{equation}
      \left( \begin{array}{c} \bs p_t\\ \bs q_t \end{array} \right) = \bs{\mathcal M}_t \left( \begin{array}{c} \bs p_0\\ \bs q_0 \end{array} \right).
    \end{equation}
    \(\bs{\mathcal M}_t^{-1} = - \bs{\mathcal J} \bs{\mathcal M}_t^T \bs{\mathcal J}\) since \(\bs{\mathcal M}_t\) is symplectic. This means that we can express the matrix inversion as follows:
    \begin{eqnarray}
      \left( \begin{array}{c} \bs p_0\\ \bs q_0 \end{array} \right) &=& \bs{\mathcal M}_t^{-1} \left( \begin{array}{c} \bs p_t\\ \bs q_t \end{array} \right)\\
      &=& - \bs{\mathcal J} \bs{\mathcal M}_t^T \bs{\mathcal J} \left( \begin{array}{c} \bs p_t\\ \bs q_t \end{array} \right)\\
      &=& - \bs{\mathcal J} \left( \begin{array}{cc}\left(\bs{\mathcal M}_t\right)_{11} & \left(\bs{\mathcal M}_t\right)_{12} \\ \left(\bs{\mathcal M}_t\right)_{21} & \left(\bs{\mathcal M}_t\right)_{22} \end{array} \right)^T \bs{\mathcal J} \left( \begin{array}{c} \bs p_t\\ \bs q_t \end{array} \right)\\
      \label{eq:stabmatinversion}
      &=& \left( \begin{array}{cc}\left(\bs{\mathcal M}_t\right)_{22} & \left(-\bs{\mathcal M}_t\right)_{12} \\ \left(-\bs{\mathcal M}_t\right)_{21} & \left(\bs{\mathcal M}_t\right)_{11} \end{array} \right) \left( \begin{array}{c} \bs p_t\\ \bs q_t \end{array} \right).
    \end{eqnarray}
    Therefore \({\left({\mathcal M}_t^{-1}\right)_{11}}_{i,j} = {\left({\mathcal M}_t\right)_{22}}_{i,j}\), and so
    \begin{equation}
      c_{ij} = {q_0}_j \cdot {q_t}_i(\bs p_0, \bs q_0) = {p_0}_j(\bs p_t, \bs q_t) \cdot {p_t}_i = {\bs \Phi_t}_{ji}.
    \end{equation}
    
    This property can also be seen in the drawing of phase space shown in Fig.~\ref{fig:qdotproduct}. There, initial perpendicular \({p_0}_j\) and \({q_0}_j\) manifolds are drawn along with harmonically evolved \({p_t}_i\) and \({q_t}_i\) manifolds, which remain perpendicular to each other and make an angle \(\alpha\) to the first \({p_0}_j\) and \({q_0}_j\) manifolds, respectively. The projection of \({q_t}_i(\bs p_0, \bs q_0)\) onto \({q_0}_j\) can be represented as the length \(b\) of a right triangle's adjacent side to the angle \(\alpha\), with an opposite side set to some length \(a\). The projection of \({p_0}_j(\bs p_t, \bs q_t)\) onto \({p_t}_i\) is similarly represented by the length \(b'\) of a right triangle's adjacent side to the angle \(\alpha\), with an opposite side also set to length \(a\). It follows that the third angle \(\beta\) in both triangles must be the same, and so by the law of sines
    \begin{equation}
      \frac{a}{\sin \alpha} = \frac{b}{\sin \beta} = \frac{b'}{\sin \beta}.
    \end{equation}
    Therefore \(b = b'\) and so these two projections are equal to one another.

    In the discrete Weyl phase space such manifolds must lie along grid phase points and obey the periodicity in \(x_p\) and \(x_q\), but the premise is the same.
\end{proof}
Overall, the procedure outlined in Lemma~\ref{lem:measurement} for deterministic measurements takes \(\mathcal O(n^2)\) operations since Eq.~\ref{eq:deterministiccoeff} is a sum of \(\mathcal O(n)\) vectors made up of \(\mathcal O(n)\) components. So, the overall measurement protocol takes \(\mathcal O(n^2)\) operations. Note that this formulation of the algorithm shows that it is the symplectic structure on phase space and the linear transformation under harmonic evolution that allows the inversion (Eq.~\ref{eq:stabmatinversion}) to be performed efficiently.

\section{Aaronson-Gottesman Tableau Algorithm}
\label{sec:aaronsongott}

The Aaronson-Gottesman tableau algorithm was originally defined for qubits (\(d = 2\))~\cite{Aaronson04}. Like the algorithm we presented in the previous section, it only requires overall \(\mathcal O(n^2)\) operations for propagationan and measurement for \(n\) qubits. The algorithm has been proven to be extendable to \(d>2\)~\cite{deBeaudrap11} and similar algorithms have been formulated in \(d>2\)~\cite{Yoder12}. Alternatives have also been developed to the tableau formalism, though they prove to be equally efficient in worst-case scenarios~\cite{Anders06}. However, we are not aware of any direct extension of the Aaronson-Gottesman tableau algorithm to dimensions greater than two. In this and the next section, we will show that the Wigner algorithm presented in Section~\ref{sec:wigstabalgorithm} is equivalent to the Aaronson-Gottesman tableau algorithm.

\subsection{Stabilizer Representation}
\label{sec:stabrep2}
The Aaronson-Gottesman algorithm is defined the stabilizer formalism. It keeps track of the evolution of a stabilizer state by updating the generators of the stabilizer group, elements of which are defined as follows:
  \begin{definition}
    \label{def:stabilizers}
    A set of operators that satisfies \(S= \{ \hat g\in\mathcal{P}\) such that \(\hat g\ket{\psi}=\ket{\psi} \} \) are called the \textit{stabilizers} of state \(\ket{\psi}\), where \(\mathcal P\) is the set of Pauli operators, each of which has the form \(e^{\frac{\pi i}{2}\alpha} \hat P_1 \cdots \hat P_n\) where \(\alpha \in \{0,1,2,3\}\) for \(n\) qubits with \(\hat P_i \in \{\hat I_i, \hat Z_i, \hat X_i, \hat Y_i\}\).
  \end{definition}

For the sake of completeness, we present here a summary of the Aaronson-Gottesman algorithm, in order to compare it to our algorithm. For more details, see~\cite{Gottesman98,Aaronson04}.

Each \(n\)-qubit stabilizer state is uniquely determined by \(2^n\) Pauli operators. There are only \(n\) generators of this Abelian group of \(2^n\) operators. Therefore, an \(n\)-qubit stabilizer state is defined by the \(n\) generators of its stabilizer state. Every element in this set of generators, \(\{\hat g_1, \hat g_2,\ldots, \hat g_n \}\), is in the Pauli group, and each generator has the form:
\begin{equation}
  \hat g_{i}=\pm \hat P_{i1}\ldots \hat P_{in}.
  \label{eq:generator}
\end{equation}
Any unitary propagation by Clifford operators or measurement of the stabilizer state changes at least some of the \(\hat P_{ij}\) elements of the \(n\) generators of the state's stabilizer. This includes the \(\pm 1\) phase in Eq.~\ref{eq:generator}, which must also be kept track of in Aaronson-Gottesman's algorithm.

\subsection{Unitary Propagation}
\label{sec:unitprop2}

For each Clifford operation, Aaronson and Gottesman showed that only \(\mathcal{O}(n)\) operations are necessary to update all generators~\cite{Aaronson04}. Specifically, according to the update rules in Table~\ref{tab:aaronsongottesmanprop}, each generator can be updated with a constant number of operators for every single Clifford gate, therefore \(O(n)\) in total. However, it is a little more complicated to update the generators after each measurement. To do this efficiently, Aaronson introduced ``\textit{destabilizers}'':
  
 \begin{definition}
 \label{def:destabilizers}
 \textit{Destabilizers} \(\{ \hat g'_1,\ldots,\hat g'_n \}\) are the operators that generate the full Pauli group with the stabilizers \(\{ \hat g_1,\ldots, \hat g_n\}\). They
 have the following properties:
 \begin{itemize}
 \item[(i)] \(\hat g'_1\), \(\hat g'_2\), \(\ldots\), \(\hat g'_n\) commute.
 \item[(ii)] Each destabilizer \(\hat g'_h\) anti-commutes with the corresponding stabilizer \(\hat g_h\), and commutes with all other stabilizers.
 \end{itemize}
 \end{definition}
  
 To incorporate the destabilizers, a \textit{tableau} becomes useful to see how they play a role in updating the stabilizer generators during measurement~\cite{Aaronson04}. 
  \begin{table}[h]
    \begin{tabular}{|c|c|c|}
      \hline
      Gates & Input & Output\\
      \hline
      \multirow{2}*{Hadamard}
      & \(\hat X\) & \(\hat Z\) \\
      & \(\hat Z\) & \(\hat X\)\\
      \hline
      \multirow{2}*{phase}
      & \(\hat X\) & \(\hat Y\) \\
      & \(\hat Z\) & \(\hat Z\)\\
      \hline
      \multirow{4}*{CNOT}
      & \(\hat X \otimes \hat I\) & \(\hat X \otimes \hat X\) \\
      & \(\hat I \otimes \hat X\) & \(\hat I \otimes \hat X\) \\
      & \(\hat Z \otimes \hat I\) & \(\hat Z \otimes \hat I\) \\
      & \(\hat I \otimes \hat Z\) & \(\hat Z \otimes \hat Z\) \\
      \hline
    \end{tabular}
    \caption{Transformation of stabilizer generators under Clifford operations.}
    \label{tab:aaronsongottesmanprop}
  \end{table}
  
 Aaronson-Gottesman defined such a \(2n\times (2n+1)\) binary tableau matrix as:
  \[
  \left(
  \begin{matrix}
    x_{11}   & \ldots & x_{1n} & \vline & z_{11}   & \ldots & z_{1n} & \vline & r_{1}\\
    \vdots    & \ddots & \vdots & \vline & \vdots  & \ddots & \vdots & \vline & \vdots\\
    x_{n1}   &\ldots & x_{nn} & \vline & z_{n1}   &\ldots & z_{nn} & \vline & r_{n}\\
    \hline
    x_{(n+1)1}   & \ldots & x_{(n+1)n} & \vline & z_{(n+1)1}   & \ldots & z_{(n+1)n} & \vline & r_{n+1}\\
    \vdots         & \ddots & \vdots        & \vline & \vdots         & \ddots & \vdots        & \vline & \vdots\\
    x_{(2n)1}   &\ldots & x_{(2n)n}       & \vline & z_{(2n)1}   &\ldots & z_{(2n)n}    & \vline & r_{2n}\\
  \end{matrix}
  \right)
  \]  
  This matrix contains \(2n\) rows. The first n rows denote the destabilizers \(\hat g'_1\) to \(\hat g'_n\) while rows \((n+1)\) to \(2n\) represent the stabilizers \(\hat g_1\) to \(\hat g_n\). The \((n+1)\)th bit in each row denotes the phase \((-1)^{r_i}\) for each generator. We encode the \(j\)th Pauli operator in the \(i\)th row as shown in Table~\ref{tab:Paulibinaryencoding}.
  \begin{table}
    \begin{tabular}{|c|c|c|}
      \hline
      \(x_{ij}\) & \(z_{ij}\) & \(\hat P_j\)\\
      \hline
      \(0\) & \(0\) & \(\hat I_j\)\\
      \(0\) & \(1\) & \(\hat Z_j\)\\
      \(1\) & \(0\) & \(\hat X_j\)\\
      \(1\) & \(1\) & \(\hat Y_j\)\\
      \hline
    \end{tabular}
    \begin{tabular}{|c|c|}
      \hline
      \(r_i\) & phase\\
      \hline
      \(0\) & \(+1\)\\
      \(1\) & \(-1\)\\
      \hline
    \end{tabular}
    \caption{Binary representation of the Pauli operators and the Pauli group phase used in their tableau representation.}
    \label{tab:Paulibinaryencoding}
  \end{table}
  We can update the stabilizers and destabilizers as follows:
  
  \emph{Hadamard gate on qubit \(i\)} For all \(j\in\{1,2,...,2n\}\), \(r_{j} \mapsto r_{j}\oplus x_{ji}z_{ji}\), then swap \(x_{ji}\) with \(z_{ji}\).
  
  \emph{Phase gate on qubit \(i\)} For all \(j\in\{1,2,...,2n\}\), \(r_{j} \mapsto r_{j}\oplus x_{ji}z_{ji}\), \(z_{ji} \mapsto z_{ji}\oplus x_{ji}\).
  
  \emph{CNOT gate on control qubit \(i\) and target qubit \(j\)} For all \(k\in\{1,2,...,2n\}\), \(r_{k} \mapsto r_{k}\oplus x_{ki}z_{kj}(x_{kj}\oplus z_{ki}\oplus 1)\), \(x_{kj} \mapsto x_{kj}\oplus x_{ki}\), \(z_{ki} \mapsto z_{ki}\oplus z_{kj}\).

  These actions correspond to those given in Table~\ref{tab:aaronsongottesmanprop}.
  
Notice the striking similarity of these tableau transformation rules under unitary propagation to the \(\bs \Phi\) transformation rules in Section~\ref{sec:aaronsongott}. The most glaring difference is that the Aaronson-Gottesman algorithm involves updates of the vector \(\bs r\). We will discuss this and its connection to the dimension \(d=2\) of the system in Section~\ref{sec:discussion}. It is clear that these transformations also take \(O(n)\) operations each.

\subsection{Measurement}
\label{sec:meas2}
To describe the measurement part of the algorithm, we need to first define a rowsum operation in the tableau which corresponds to multiplying two Pauli operators together.

As defined in~\cite{Aaronson04}:\\
\textbf{Rowsum}: To sum row \(i\) and \(j\), first update the bits that represent operators by \(x_{ik}\oplus x_{jk}\) and
\(z_{ik}\oplus z_{jk}\) for \(k=1,\,\ldots,\,n\). To calculate the resultant phase, Aaronson and Gottesman first defined the following function:
\begin{equation}
  f(x_{ik},x_{jk},z_{ik},z_{jk})=
    \begin{cases}
      0 &\text{if } x_{ik}=z_{ik}=0,\\
      z_{jk}-x_{jk} &\text{if } x_{ik}=z_{ik}=1, \\
      z_{jk}(2x_{jk}-1) &\text{if } x_{ik}=0, z_{ik}=1, \\
      x_{jk}(1-2z_{jk}) &\text{if } x_{ik}=1, z_{ik}=0.
    \end{cases}
\end{equation}
Since each stabilizer generator is the tensor product of \(n\) single qubit Pauli operators (see Eq.~\ref{eq:generator}), they must be multiplied together to obtain the phase:
\begin{equation}
  \begin{cases}
    0 &\text{if } r_{i}+r_{j}+\sum_{k=1}^{n} f(x_{ik},x_{jk},z_{ik},z_{jk})\equiv 0~(\text{mod}~4),\\
      1 &\text{if } r_{i}+r_{j}+\sum_{k=1}^{n} f(x_{ik},x_{jk},z_{ik},z_{jk})\equiv 2~(\text{mod}~4).
  \end{cases}
\end{equation}

Having defined the rowsum function, let us now consider a measurement of \(\hat Z_i\) on qubit \(i\). For \(d=2\), Pauli group operators can only commute or anti-commute with each other. If \(\hat Z_i\) anti-commutes with one or more of the generators then the measurement is random. If \(\hat Z_i\) commutes with all of the generators then the measurement is deterministic. We consider these two cases:
  
  \textbf{Case 1:} Random Measurement

  \(\hat Z_i\) anti-commutes with one or more of the generators. If there is more than one, we can always pick a single anti-commuting generator, \(\hat g_j\), and update the rest by replacing them with their product with \(\hat g_j\) (i.e. taking the rowsum of their corresponding rows) such that they commute with \(\hat Z_i\). These updates take \(\mathcal{O}(n^2)\) operations. Finally, we only need to replace \(\hat g_j\) by \(\hat Z_i\).

  In other words, with respect to the tableau, there should exist at least one \(j\in\{n+1,n+2,...,2n\}\) such that \(x_{ji}=1\). Replacing all rows where \(x_{ki}=1\) for \(k\neq j\) with the sum of the \(j\)th and \(k\)th row (using the rowsum function) sets all \(x_{ki}=0\) for \(k\neq j\).

  Finally, we replace the \((j-n)\)th row with the \(j\)th row and update the \(j\)th row by setting \(z_{ji}=1\) and all other \(x_{jk}\)s and \(z_{jk}\)s to \(0\) for all \(k\). We output \(r_{j}=0\) or \(r_j=1\) with equal probability for the measurement result. This procedure takes \(O(n^2)\) operations because each rowsum operation takes \(O(n)\) operations and up to \(n-1\) rowsums may be necessary.
  
  \textbf{Case 2:} Deterministic Measurement

  \(\hat Z_i\) commutes with all generators. In this case, there is no \(j\in\{n+1,n+2,...,2n\}\) such that \(x_{ji}=1\) and we don't need to update any of the generators. But we do need to do some work to retrieve the measurement outcome. 

  Measurement \(\hat Z_i\) commutes with all of the stabilizers, therefore either \(+\hat Z_i\) or \(-\hat Z_i\) is a stabilizer of the state. So it must be generated by the generators. The sign \(\pm 1\) is the measurement outcome we are looking for. This means that
  \begin{equation}
    \prod_{j=1}^n \hat g_{j}^{c_j}=\pm \hat Z_i,
  \end{equation}
  where \(c_j=1\) or \(0\).
  
  For those destabilizers \(g'_k\) that satisfy
  \begin{equation}
    \label{eq:anticommwithg}
    \{\hat g'_k, \pm \hat Z_i\}=0,
  \end{equation}
  \(c_k=1\). Otherwise \(c_k=0\).
  This can be seen from
  \begin{equation}
    \{\hat g'_k, \pm \hat Z_i\}=\{\hat g'_k, \prod_{j=1}^{n} \hat g_j^{c_j}\}=\prod_{\substack{j=1\\j\ne k}}^n \hat g_j^{c_j} \{\hat g'_k,\hat g_k^{c_k}\}=0,
  \end{equation}
  where we used part (ii) of Definition~\ref{def:destabilizers} of the destabilizers and Eq.~\ref{eq:anticommwithg}.
  The last equality requires \(c_k=1\).

  Therefore, to find the deterministic measurement outcome, the stabilizers whose corresponding destabilizer anti-commutes with the measurement operation \(\hat Z_i\) must be multiplied together. Every row \((n+j)\) in the bottom half of the tableau, such that \(x_{ji}=1\) (for \(j \in \{1,\,\ldots,\,n\}\)), can be added up together and stored in a temporary register. The resultant phase \(\pm 1\) of this sum is the measurement result we are looking for.

  Checking if each destabilizer commutes or anti-commutes with \(\hat Z_i\) takes a constant number of operations. One multiplication takes \(O(n)\) operations, and there are \(O(n)\) multiplications needed. Therefore, a measurement takes \(O(n^2)\) operations overall.

\section{Discussion}
\label{sec:discussion}

As we made clear throughout Section~\ref{sec:aaronsongott}, the scaling of the number of required operations with respect to number of qubits \(n\) is exactly the same in the Aaronson-Gottesman algorithm in the Wigner algorithm presented in Section~\ref{sec:wigstabalgorithm}. The two algorithms also require the same number of dits of temporary storage for performing the deterministic measurement. Moreover, there is a correspondence between the tableau employed by Aaronson-Gottesman and the matrix \(\bs \Phi_t\) and vector \(\bs r_t\) we use. In particular, the tableau is equal to \(\left( \begin{array}{c|c} \bs \Phi_t & \bs r_t \end{array}\right)\):
\begin{widetext}
\begin{equation}
\bs \Phi_t = \left(\begin{array}{c|c}\partder{\bs p_0}{\bs p_t} & \partder{\bs p_0}{\bs q_t}\\ \hline \partder{\bs q_0}{\bs p_t} & \partder{\bs q_0}{\bs q_t}\end{array}\right) \equiv   \left(
  \begin{matrix}
    x_{11}   & \ldots & x_{1n} & \vline & z_{11}   & \ldots & z_{1n}\\
    \vdots    & \ddots & \vdots & \vline & \vdots  & \ddots & \vdots\\
    x_{n1}   &\ldots & x_{nn} & \vline & z_{n1}   &\ldots & z_{nn}\\
    \hline
    x_{(n+1)1}   & \ldots & x_{(n+1)n} & \vline & z_{(n+1)1}   & \ldots & z_{(n+1)n}\\
    \vdots         & \ddots & \vdots        & \vline & \vdots         & \ddots & \vdots\\
    x_{(2n)1}   &\ldots & x_{(2n)n}       & \vline & z_{(2n)1}   &\ldots & z_{(2n)n}\\
  \end{matrix}
  \right)
\end{equation}
\end{widetext}
and
\begin{equation}
  \bs r_t = \left(\begin{array}{c} \bs r_p \\ \hline \bs r_q \end{array}\right) \equiv \left(\begin{array}{c} r_1 \\ \vdots \\ r_n\\\hline r_{n+1} \\ \vdots \\ r_{2n} \end{array}\right).
\end{equation}
This can be seen through the following equation:
\begin{eqnarray}
\exp\left( \frac{2 \pi i}{d} \sum_{j=1}^{2n} {\Phi_t}_{n+i,j} \hat x_j \right) \Psi_t &=& \prod_{j=1}^{2n} \exp\left( \frac{2 \pi i}{d} {\Phi_t}_{n+i,j} \hat x_j \right) \Psi_t \nonumber
  \label{eq:heisenbergeqschrodinger1}\\
  &=& \exp\left(\frac{2 \pi i}{d} {r_t}_i \right) \Psi_t
\end{eqnarray}
where \(\hat{\bs x} \equiv \left( \hat{\bs p}, \hat{\bs q}\right)\). Multiplying the right hand side of the first equation and the second equation by \(\exp\left(-\frac{2 \pi i}{d} {r_t}_i \right)\), it follows that
\begin{equation}
  \label{eq:heisenbergeqschrodinger2}
  \exp\left(-\frac{2 \pi i}{d} {r_t}_i \right) \prod_{j=1}^{2n} \exp\left( \frac{2 \pi i}{d} {\Phi_t}_{n+i,j} \hat x_j \right) \Psi_t = \hat g_i \Psi_t = \Psi_t.
\end{equation}
In other words, \({r_t}_i\) specifies the phase \(\exp\left(-\frac{2 \pi i}{d} {r_t}_i \right)\) of the \(i\)th stabilizer, which is itself specified by \({\Phi_t}_{n+i,j}\) for \(j \in \{0, \ldots, 2n\}\). These are the same roles for \(\bs r\) and the tableau in the Aaronson-Gottesman tableau algorithm~\cite{Aaronson04}.

Indeed, both algorithms check the bottom half of their matrices for finite elements of \(\bs \Phi_{n+j,i}\) to determine if a measurement on the \(i\)th qudit will be random or not. They also use a very similar protocal to determine the outcome of deterministic measurements. The Wigner-based algorithm motivates these manipulations in terms of the symplectic structure of Weyl phase space and the relationship between the two Wigner functions specified by the top and bottom of \(\bs \Phi\), providing a strong physical intuition for their effects. Aaronson and Gottesman motivate these manipulations using the anti-commutation relations between the stabilizer and destabilizer generators. Also, the latter half of both the Wigner function's \(\bs r_t\) and Aaronson-Gottesman's \(\bs r\) are used to determine measurement outcomes. The only fundamental algorithmic difference between the approaches is that the Wigner-based algorithm does not require updates of \(\bs r_t\) during unitary propagation. The reason for this lies in the fact that Aaronson-Gottesman's algorithm deals with systems with \(d=2\) while the Wigner-based algorithm is restricted to odd \(d\).

In particular, for the one-qubit Clifford group gate operator \(\hat A = \{\hat P_i, \hat F_i\} \, \forall i = \{1,\ldots,n\}\), the Aaronson-Gottesman algorithm specifies that for a \(q\)- or \(p\)-state, its Wigner function evolves by:
\begin{equation}
  \Psi_x(\bs{\mathcal M}_{\hat A} \bs x).
\end{equation}
But for \(\ket r = \frac{1}{\sqrt{2}} (\ket 0 \pm i \ket 1)\), a \(Y\)-state which is diagonal to \(q\) and \(p\), its Wigner function must first be translated:
\begin{equation}
  \Psi_x\left(\bs{\mathcal M}_{\hat A} \left(\bs x + \bs \beta \right) \right),
\end{equation}
where the translation \(\bs \beta\) can be \(( 1, 0)\) or \((0, 1)\) equivalently. There is a similar state-dependence for the two-qudit CNOT gate \(\hat C_{ij}\).

This demonstrates that the Aaronson-Gottesman algorithm is state-dependent on the qubit stabilizer state it is acting on. On the other hand, our algorithm on qudit stabilizer states is state-independent. This likely is a consequence of the fact that the Weyl-Heisenberg group, which is made up of the boost and shift operators defined in Eqs.~\ref{eq:boost} and~\ref{eq:shift} that underlie the discrete Wigner formulation, are a subgroup of \(U(d)\) instead of \(SU(d)\) for \(d=2\)~\cite{Bengtsson17}. Furthermore, qubits exhibit state-independent contextuality while odd \(d\) qudits do not~\cite{Mermin93}.

\section{Example of Stabilizer Evolution}

\begin{figure}[ht]
\includegraphics[scale=1.0]{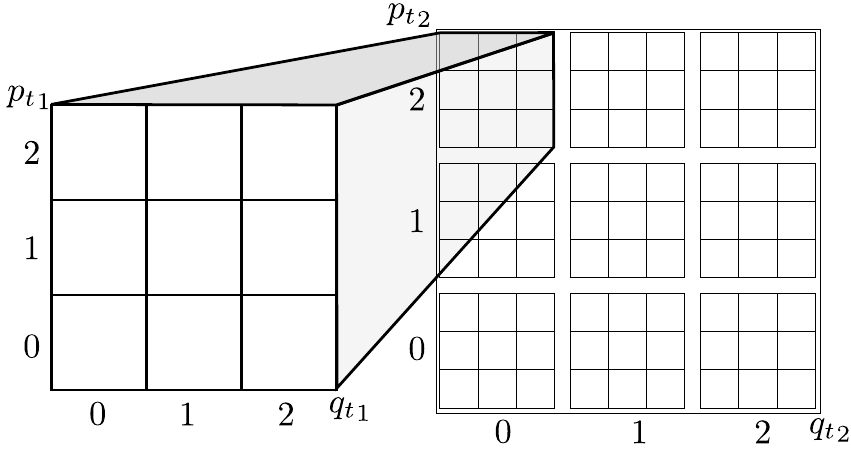}
\caption{A decomposition of the two qutrit Wigner function into nine \(3 \times 3\) girds, where each \(3 \times 3\) grid denotes the value of the Wigner function at all \({p_t}_1\) and \({q_t}_1\) for a fixed value of \({p_t}_2\) and \({q_t}_2\) denoted by the external axes. This organization is used in Fig.~\ref{fig:twoqutritquantum} below.}
\label{fig:twoqutritquantumsetup}
\end{figure}
\begin{figure}[ht]
\includegraphics[scale=1.0]{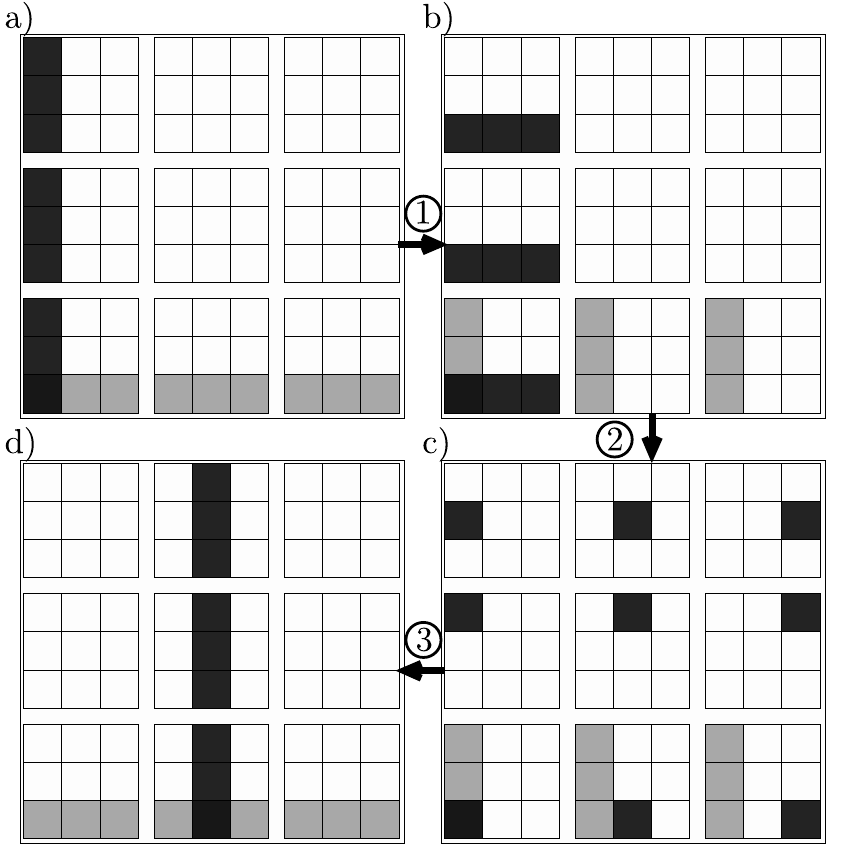}
\caption{The Wigner function of two qutrits initially prepared in a) the state \(\ket{0} \otimes \ket{0}\). (1) It is then evolved under \(\hat F_1\) to produce b) \(\frac{1}{\sqrt{3}}\left(\ket 0 + \ket 1 + \ket 2\right) \otimes \ket 0\). (2) This state is subsequently evolved under \(\hat C_{12}\) producing c) the Bell state \(\frac{1}{\sqrt{3}}\left(\ket{00} + \ket{11} + \ket{22}\right)\). (3) Qudit \(1\) is then measured producing the random outcome \(1\), which collapses qudit \(2\) into the same state, so that d) \(\ket 1 \otimes \ket 1\) results. The black color indicates the Wigner function specified by the lowest \(n\) rows of \(\delta_{\bs \Phi_t \bs x, \bs r_t}\), and the gray color indicates the Wigner function specified by the highest \(n\) rows (\(\bs q_0(\bs p_t, \bs q_t)\) and \(\bs p_0(\bs p_t, \bs q_t)\), respectively). The evolution and algorithmic implementation are explained in the text.}
\label{fig:twoqutritquantum}
\end{figure}

As a demonstration of what stabilizer state propagation looks like in the Wigner formalism, we proceed to go through an example of Bell state preparation and measurement starting from the state \(\ket 0 \otimes \ket 0\). The prepared Wigner function is illustrated in Fig.~\ref{fig:twoqutritquantum} with the color black, and the Wigner function represented by setting \(\bs \Phi_0 = \left(\begin{array}{cc}1 & 0\\ 0& 0 \end{array}\right)\) (i.e. considering the top \(n\) rows of our function as a separate Wigner function) is illustrated with the color gray.

We begin with
\begin{eqnarray}
  &&\Psi_x(\bs x) = \\
  &&\delta_{\left(\begin{array}{cccc} 1 & 0 & 0 & 0\\ 0 & 1 & 0 & 0\\ 0 & 0 & 1 & 0\\ 0 & 0 & 0 & 1\end{array}\right) \left(\begin{array}{c}{p_t}_1\\{p_t}_2\\{q_t}_1\\{q_t}_2\end{array}\right),\left(\begin{array}{c}0\\0\\0\\0\end{array}\right)} = \delta_{\left(\begin{array}{c}{p_t}_1\\{p_t}_2\\{q_t}_1\\{q_t}_2\end{array}\right),\left(\begin{array}{c}0\\0\\0\\0\end{array}\right)}, \nonumber
\end{eqnarray}
denoting an initially prepared state of \(\ket{0} \otimes \ket{0}\). This is clear in Figure~\ref{fig:twoqutritquantum}a by the black band that lies along all Weyl phase space points with \({q_t}_1=0\) and \({q_t}_2=0\). On the other hand, the gray manifold is perpendicular to the black one, and lies along Weyl phase space points with \({p_t}_1=0\) and \({p_t}_2=0\).

Acting on this state with \(\hat F_1\) produces \(\frac{1}{\sqrt{3}} \left(e^{\frac{2 \pi i}{3} 0 \times 0} \ket 0 + e^{\frac{2 \pi i}{3} 1 \times 0} \ket 1 + e^{\frac{2 \pi i}{3} 2 \times 0} \ket 2 \right) \otimes \ket{0}\). Applying the algorithm specified at the end of Section~\ref{sec:propalgorithm}, we find:
\begin{eqnarray}
  &&\Psi_x(\bs x) = \\
  &&\delta_{\left(\begin{array}{cccc} 0 & 0 & -1 & 0\\ 0 & 1 & 0 & 0\\ 1 & 0 & 0 & 0\\ 0 & 0 & 0 & 1\end{array}\right) \left(\begin{array}{c}{p_t}_1\\{p_t}_2\\{q_t}_1\\{q_t}_2\end{array}\right),\left(\begin{array}{c}0\\0\\0\\0\end{array}\right)} = \delta_{\left(\begin{array}{c}-{q_t}_1\\{p_t}_2\\{p_t}_1\\{q_t}_2\end{array}\right),\left(\begin{array}{c}0\\0\\0\\0\end{array}\right)}. \nonumber
\end{eqnarray}
Thus, the momentum of qudit \(1\) is now determined and is \(0\) while the second qudit is unchanged. This can be seen in Fig.~\ref{fig:twoqutritquantum}b, where the \({q_t}_2\) values of the non-zero Weyl phase space points are the same, while the state has rotated by \(-\pi/2\) in \(({p_t}_1,{q_t}_1)\)-space. A similar transformation has occured for the perpendicular gray manifold.

Acting next with \(\hat C_{12}\) produces the Bell state \(\frac{1}{\sqrt{3}}\left(\ket{00} + \ket{11} + \ket{22}\right)\), which is represented by the following Wigner function:
\begin{eqnarray}
  &&\Psi_x(\bs x) = \\
  &&\delta_{\left(\begin{array}{cccc} 0 & 0 & -1 & 0\\ 0 & 1 & 0 & 0\\ 1 & 1 & 0 & 0\\ 0 & 0 & -1 & 1\end{array}\right) \left(\begin{array}{c}{p_t}_1\\{p_t}_2\\{q_t}_1\\{q_t}_2\end{array}\right),\left(\begin{array}{c}0\\0\\0\\0\end{array}\right)} = \delta_{\left(\begin{array}{c}-{q_t}_1\\{p_t}_2\\{p_t}_1+{p_t}_2\\-{q_t}_1+{q_t}_2\end{array}\right),\left(\begin{array}{c}0\\0\\0\\0\end{array}\right)}. \nonumber
\end{eqnarray}
The entanglement between the two qudits is evident in both of their dependence on each other's momenta and positions, \({p_t}_1 = -{p_t}_2\) and \({q_t}_1 = {q_t}_2\), specified by the last two rows. Fig.~\ref{fig:twoqutritquantum}c shows that the state is still representable as lines in Weyl phase space, except they now traverse through the different planes of \(({q_t}_1,{p_t}_1)\) associated with each value of \(({q_t}_2,{p_t}_2)\). However, if you consider the left column in Fig.~\ref{fig:twoqutritquantum}c corresponding to \({q_t}_2=0\), you can see that the only black Weyl phase points are at \({q_t}_1=0\). Similarly, the middle column corresponding to \({q_t}_2=1\) shows that \({q_t}_1=1\), and the right column corresponding to \({q_t}_2=2\) shows that \({q_t}_1=2\) too, confirming that \(\ket Phi = \frac{1}{\sqrt{3}}\left(\ket{00} + \ket{11} + \ket{22}\right)\). Thus, the entanglement of the two qudits positions is clearly evident in this figure of the Wigner function.

We then proceed to measure qudit \(1\). Since the lower two equations involve \({p_t}_1\), we know that this is a random measurement. Let us pick the outcome to be \(1\) and set the third row as such, replacing the first row with the old third row. This collapses qudit \(2\) into the same state:
\begin{eqnarray}
  &&\Psi_x(\bs x) =\\
  &&\delta_{\left(\begin{array}{cccc} 1 & 1 & 0 & 0\\ 0 & 1 & 0 & 0\\ 0 & 0 & 1 & 0\\ 0 & 0 & -1 & 1\end{array}\right) \left(\begin{array}{c}{p_t}_1\\{p_t}_2\\{q_t}_1\\{q_t}_2\end{array}\right),\left(\begin{array}{c}0\\0\\1\\0\end{array}\right)} = \delta_{\left(\begin{array}{c}{p_t}_1+{p_t}_2\\{p_t}_2\\{q_t}_1\\-{q_t}_1+{q_t}_2\end{array}\right),\left(\begin{array}{c}0\\0\\1\\0\end{array}\right)}. \nonumber
\end{eqnarray}
The lower two rows show that now \({q_t}_1 = 1\), as we chose, and \({q_t}_2 = {q_t}_1 = 1\). The collapse of qudit \(2\) into \(\ket 1\) can also been seen in Fig.~\ref{fig:twoqutritquantum}c by the fact that \({q_t}_1 = 1\) only in the \(3\times3\) grids that correspond to \({q_t}_2 = 1\) too.

Finally, the fact that the measurement of \({q_t}_2\) is deterministic can be seen in that \({p_t}_2\) is not present in the last two rows of \(\bs \Phi_t\). Furthermore, it is clear since the first row has a coefficient of \(1\) in front of \({p_t}_1\), that the corresponding third row must be added with weight \(1\) to the fourth row to obtain this deterministic measurement outcome of \({q_t}_2 = 1\). This can also be seen in Fig.~\ref{fig:twoqutritquantum} by finding the projection of \({p_0}_1\) onto \({p_t}_2\), which are shown by the gray manifolds in panels a) and d) respectively. They are collinear and so the projection is equal to \(1\). (Perpendicular manifolds corresponds to a projection of \(0\), and those that lie \(\pi/4\) diagonally with respect to each other have a projection equal to \(2\) in this discrete geometry.)

\section{Conclusion}
\label{sec:conc}
In summary, we introduced an algorithm that efficiently simulates stabilizer state evolution under Clifford gates and measurements in the \(\hat Z\) Pauli basis for odd \(d\) qudits. We accomplished this by relying on the phase-space perspective of stabilizer states as discrete Gaussians and Clifford operators as having underlying harmonic Hamiltonians. We showed the equivalence of our algorithm, through Eqs.~\ref{eq:heisenbergeqschrodinger1} and \ref{eq:heisenbergeqschrodinger2}, to the well-known Aaronson-Gottesman tableau algorithm~\cite{Aaronson04}, revealing that Aaronson-Gottesman's tableau corresponds to a discrete Wigner function. As a consequence, we revealed the physically intuitive phase space perspective of Aaronson-Gottesman's algorithm, as well as its extension to higher odd \(d\). 

This work illustrates that no efficiency advantage is gained by using the Heisenberg representation for stabilizer propagation. Eq.~\ref{eq:heisenbergeqschrodinger2} indicates that the Heisenberg representation is equivalent to the Schr\"odinger representation in this context; evolving the operators is just as efficient as evolving the states, as perhaps expected.

Lastly, the correspondence between the Wigner-based algorithm and the Aaronson-Gottesman tableau algorithm may point the direction on how to resolve the long-standing issue of describing the Wigner-Weyl-Moyal and center-chord formalism for \(d=2\) systems. We have shown that the Aaronson-Gottesman algorithm is essentially a \(d=2\) treatment of the Wigner approach. The salient ingredient appears to be the state-dependence of this evolution, and likely is related to the fact that the Weyl-Heisenberg group is not in \(SU(d)\) for \(d=2\) and that qubits exhibit state-independent contextuality, which odd \(d\) qudits do not. Exploring the details of this state-dependence is a promising subject of future study.

\bibliography{biblio}{}
\bibliographystyle{unsrt}

\end{document}